\documentclass[11pt]{article}
\usepackage{geometry}
\usepackage{amsmath,amssymb,amsthm,mathtools}
\usepackage{enumitem}
\usepackage{xcolor}
\usepackage{hyperref}
\usepackage[utf8]{inputenc}
\usepackage[T1]{fontenc}
\usepackage{caption}
\usepackage{color}
\usepackage{hyperref}
\usepackage{bookman}
\usepackage[T1]{fontenc}
\usepackage{tikz}
\usepackage{multirow}
\usepackage{dirtytalk}
\usepackage{lscape}
\usepackage{array} 
\usepackage{lmodern}
\usepackage{xcolor}
\usepackage{setspace}
\makeatletter
\newcommand{\mathleft}{\@fleqntrue\@mathmargin0pt}
\makeatother
\usepackage{nccmath} 
\usepackage{graphicx}
\usepackage{eurosym}
\usepackage{comment}
\usepackage{natbib}
\usepackage{booktabs}

\definecolor{craneorange}{RGB}{0,102,131}
\definecolor{craneblue}{RGB}{0,0,0}
\definecolor{ticketswapblue}{RGB}{0,174,221}

\hypersetup{
    colorlinks=true,
    linkcolor=blue!60!black,
    citecolor=blue!60!black,
    urlcolor=blue!60!black
}

\newtheorem{definition}{Definition}
\newtheorem{remark}{Remark}

\newtheorem{theorem}{Theorem}
\newtheorem{example}{Example}

\newcommand{\Prefs}{\mathcal P}
\newcommand{\Allocs}{\mathcal X}
\newcommand{\Supergroups}{L}
\newcommand{\set}[1]{\left\{#1\right\}}
\newcommand{\Endowment}{e}
\title{Non-obvious Manipulability with Groups in Shapley--Scarf Housing Markets}
\author{Louise Demoor\footnote{CREST, Ecole Polytechnique, Institut Polytechnique de Paris, email: \href{mailto:louise.demoor@polytechnique.edu}{louise.demoor@polytechnique.edu}}
\and Martí Jané-Ballarín\footnote{University of Barcelona, 
email: \href{mailto:marti.jane@ub.edu}{marti.jane@ub.edu}}
\and Pierre Nunn\footnote{Université Jean Monnet, Saint-Étienne, GATE-LSE, email: \href{mailto:pierre.nunn@univ-st-etienne.fr}{pierre.nunn@univ-st-etienne.fr}}
\and Subhajit Pramanik\footnote{University of Padova, 
email: \href{mailto:subhajit.pramanik@phd.unipd.it}{subhajit.pramanik@phd.unipd.it}}
\and Antoine Prévotat\footnote{Université Jean Monnet, Saint-Étienne, GATE-LSE, email: \href{mailto:antoine.prevotat@univ-st-etienne.fr} {antoine.prevotat@univ-st-etienne.fr}}
\and Makoto Yokoo\footnote{Kyushu University, email:  \href{mailto:yokoo@kyudai.jp}{yokoo@kyudai.jp}}}

\begin{document}
\maketitle

\begin{abstract}
In Shapley--Scarf housing markets, \cite{ma1994} shows that top trading cycles (TTC) is the unique mechanism satisfying individual rationality (IR), Pareto efficiency (PE), and strategy-proofness. We ask what other mechanisms become possible when strategy-proofness is replaced by a weaker condition called non-obvious manipulability (NOM), introduced by \cite{troyanmorrill2020}. 
We first show that this weaker condition does not help on its own: every IR and PE mechanism is already NOM. We therefore introduce a new condition: NOM with groups, under which each agent knows the preferences of the other members of her group, but not those of agents outside the group. 
This condition reduces to strategy-proofness when all agents belong to one group, and to standard NOM when every group is a singleton. 
Also, we introduce a participation condition called \emph{group rationality} (GR), which requires that no group do worse than it would by trading only among its own members. 
We then define a class of mechanisms called \emph{TTC with super-groups},
whose members satisfy GR, PE, and NOM with groups. The class includes mechanisms that differ from 
standard TTC, including mechanisms that are not strategy-proof.
Furthermore, we show that every mechanism that
satisfies GR, PE, and NOM with groups
has the same best- and worst-case outcomes as every TTC with super-groups mechanism. 
\end{abstract}

\newpage
\section{Introduction}
\label{sec:introduction}

The Shapley--Scarf housing market is a canonical model of exchange without money, in which agents initially own indivisible houses that are reallocated according to their preferences \citep{shapley1974}. Its central mechanism, top trading cycles (TTC), selects the unique core allocation, is strategy-proof \citep{roth1982}, and, by \citet{ma1994}, is the unique mechanism satisfying individual rationality (IR), Pareto efficiency (PE), and strategy-proofness. This uniqueness leaves little room for mechanism variation if one insists on strategy-proofness. A recent line of work asks how far Ma's axioms can be relaxed: \citet{ekici2024,ekicisethuraman2024} weaken PE to pair efficiency, which rules out only mutually beneficial bilateral swaps, and TTC remains the unique rule. We instead retain IR and PE and weaken the incentive requirement.

 In many applications, strategy-proofness is stronger than what is behaviorally or informationally justified. Agents may not know the preferences of everyone in a large market, nor be able to reason through all possible reports of all other agents. This concern motivates the literature on simplicity for boundedly rational agents: \citet{li2017} argues that agents may be unable to reason through all profiles of others' reports and introduces obvious strategy-proofness, while \citet{troyanmorrill2020} move in the opposite direction with non-obvious manipulability (NOM), under which a manipulation is obvious only if it improves the agent's best- or worst-case outcome as others' reports vary. The idea is that an agent with limited information can readily identify only manipulations that improve one of these extreme outcomes.

Our first result shows that, in Shapley--Scarf markets, NOM is too weak: every IR and PE mechanism is already NOM (Theorem~\ref{thm:irpe-nom}). The endowment pins down the truthful worst case, while PE makes the favorite house attainable in the best case, so a misreport can improve neither extreme. This is peculiar to markets with endowments; in assignment problems without endowments, NOM does bind the designer, and \citet{troyan2024} finds that rank-minimizing rules are NOM only under an extra condition. Replacing strategy-proofness by NOM alone therefore imposes no additional restriction on our domain, and a stronger notion is needed.

We therefore introduce \emph{NOM with groups}. Agents are partitioned into exogenously given groups. When agent $i$ belongs to group $g$, the best and worst outcomes used to evaluate a possible misreport are computed by holding the reports of the other agents in $g$ fixed and varying only the reports of agents outside $g$. This condition interpolates between familiar requirements. If 
all agents belong to a single group, then there are no outside reports to vary, so NOM with groups becomes strategy-proofness. If every group is a singleton, it becomes ordinary NOM. For intermediate partitions, it captures a form of non-obvious manipulability under local information. It should be distinguished from group-based concepts involving coordinated deviations, such as group strategy-proofness, which rules out profitable joint misreports and also characterizes TTC \citep{bird1984,takamiya2001}, or the obvious group manipulations of \citet{bonifacio2026groups}. In our model, agents deviate individually: the group is informational, determining which other agents' reports are treated as known and fixed.

We also introduce \emph{group rationality} (GR): each agent must receive a house weakly better than her assignment in the core allocation of her own group. GR is motivated by kidney exchange, a natural instance of the Shapley--Scarf market in which a patient--donor pair is an agent and its incompatible donor is the endowment \citep{rsu2004}. As exchange programs expanded across hospitals and countries, hospitals began withholding their easy-to-match pairs and sharing only hard cases. \citet{ashlagiroth2014} show that full participation can be restored only if each hospital is guaranteed at least what it could achieve alone, and \citet{afkp2015} design a mechanism giving hospitals the incentive to report all their pairs. GR is the theoretical counterpart of that guarantee: no group is left worse off than trading only among its own members, so a hospital or country never loses by joining the global pool.

Finally, we identify mechanisms that satisfy the above intermediate incentive condition as well 
as the new participation condition. We define \emph{TTC with super-groups}. A super-group is a union of base groups, and the mechanism repeatedly applies TTC to minimal not-yet-selected super-groups, with each base group and the grand coalition included. This staged structure resembles the hierarchical and segmented exchange rules of \citet{papai2000,papai2003}, though our super-groups represent successive levels of market integration and are evaluated through GR and NOM with groups rather than full strategy-proofness. We show that every TTC with super-groups mechanism satisfies GR, PE, and NOM with groups (Theorem \ref{thm:ttc-with-super-groups}). These mechanisms illustrate that weakening strategy-proofness to NOM with groups makes room for mechanisms beyond standard TTC while preserving the basic efficiency property, 
as well as a stronger participation property. 
We also show that all TTC with super-groups mechanisms are envelope-equivalent (Theorem \ref{thm:envelope}): for every agent and within-group profile, they induce the same best and worst possible houses. Our main result gives the converse at the level of envelopes: every mechanism satisfying GR, PE, and NOM with groups induces this same envelope (Theorem \ref{thm:characterization}). Thus, the axioms characterize the extreme outcomes relevant for NOM with groups, though not the full allocation rule.

The rest of the paper is organized as follows. 
Section~\ref{sec:model} introduces the model, Section~\ref{sec:NOM-groups} develops NOM with groups, and 
Section~\ref{sec:supergroups} introduces TTC with super-groups and establishes the envelope characterization.

\section{Model}
\label{sec:model}
There is a finite set of agents $N=\set{1,2,\ldots,n}$ and a finite set of objects, which we call houses, $H$, with $|H|=|N|$. Each agent $i\in N$ is initially endowed with a house $\Endowment_i\in H$, where the initial endowment map $\Endowment:N\to H$ is a bijection. Every agent $i$ has a strict preference relation $\succ_i$ over houses, and we let $\Prefs$ denote the set of all strict preferences over $H$. We write $h\succeq_i h'$ if either $h\succ_i h'$ or $h=h'$ holds. A preference profile is a vector $\succ\,=(\succ_i)_{i\in N}\in\Prefs^N$. For a subset $S\subseteq N$ we write $\succ_S\,=(\succ_i)_{i\in S}$, and we let $\succ_{-i}$ denote a preference profile of all agents except $i$.

An allocation is a bijection $x:N\to H$; the house assigned to agent $i$ under allocation $x$ is denoted by $x_i$. Let $\Allocs$ denote the set of all allocations. A deterministic mechanism is a function $f:\Prefs^N\to\Allocs$, and $f_i(\succ)$ denotes the house assigned to agent $i$ at declared preference profile $\succ$.

Throughout, we will be concerned with the set of houses an agent has an opportunity to obtain as the reports of the remaining agents vary. For a mechanism $f$ and a declared preference $\succ_i$, let
\[
    O^f_i(\succ_i)=\set{f_i(\succ_i,\succ_{-i}) \mid \succ_{-i}\,\in\Prefs^{N\setminus\set{i}}}
\]
denote this set. Given a preference $\succ'_i$ used to rank its elements, let $b^f_i(\succ'_i;\succ_i)$ denote the best house in $O^f_i(\succ_i)$ based on $\succ'_i$; that is, $b^f_i(\succ'_i;\succ_i)=h$ such that $h\in O^f_i(\succ_i)$ and $h\succ'_i h'$ holds for any $h'\in O^f_i(\succ_i)\setminus\set{h}$. Similarly, let $w^f_i(\succ'_i;\succ_i)$ denote the worst house in $O^f_i(\succ_i)$ based on $\succ'_i$; that is, $w^f_i(\succ'_i;\succ_i)=h$ such that $h\in O^f_i(\succ_i)$ and $h'\succ'_i h$ holds for any $h'\in O^f_i(\succ_i)\setminus\set{h}$.

\subsection{Basic Axioms}

\begin{definition}[Individual Rationality (IR)]
We say house $h$ is acceptable for agent $i$, if 
$h \succeq_i \Endowment_i$ holds. 
A mechanism $f$ satisfies individual rationality (IR) if
$f_i(\succ)$ is acceptable for $i$ for every $\succ\in\mathcal P^N$
and every $i\in N$.
\end{definition}

\begin{definition}[Pareto Efficiency (PE)]
An allocation $x\in\mathcal X$ is Pareto dominated at profile $\succ$ if there exists
an allocation $y\in\mathcal X$ such that
\[
    y_i\succeq_i x_i \quad \text{for all } i\in N,
\]
and
\[
    y_j\succ_j x_j \quad \text{for some } j\in N.
\]
A mechanism $f$ satisfies Pareto efficiency (PE) if, for every profile
$\succ\in\mathcal P^N$, $f(\succ)$ is not Pareto dominated at $\succ$.
\end{definition}
\begin{definition}[Strategy-proofness (SP)]
We say that a mechanism $f$ is strategy-proof (SP)
if for every $i \in N$, every $\succ_i \in \Prefs$, and every $\succ_{-i} \in \Prefs^{N-1}$, 
there is no possible deviation $\succ'_{i}\in \Prefs$ such that $f_i\left(\succ_{i}',\,\succ_{-i}\right)
\succ_{i}f_i\left(\succ_{i}, \succ_{-i}\right)$.    
\end{definition}

\begin{definition}[Non-obvious manipulability (NOM)]
      A manipulation $\succ_{i}'$ is \emph{obvious} if, whatever the other agents may report as their preferences, either the best or worst outcomes are strictly better for $i$ than if $i$ had reported her true preference $\succ_{i}$. 
      We say mechanism $f$ is NOM if there exists no obvious manipulation. Formally, mechanism $f$ is NOM if, for any $i\in N$, $\succ_i \in \Prefs$, and $\succ'_i \in \Prefs$, both $b^f_i(\succ_i;\succ_i) \succeq_i b^f_i(\succ_i;\succ'_i)$ and $w^f_i(\succ_i;\succ_i) \succeq_i w^f_i(\succ_i;\succ'_i)$ hold. 
\end{definition}

\section{Non-obvious Manipulability with Groups}
\label{sec:NOM-groups}
Our first theorem shows that IR and PE imply NOM. 

\begin{theorem}
\label{thm:irpe-nom}
Any mechanism $f$ that is IR and PE is NOM.
\end{theorem}
\begin{proof}
    Fix any $i \in N$ and any $\succ_i, \succ'_i \in \Prefs$.
    Let $h^*$ denote $i$'s best house under $\succ_i$. 
    If $h^*=\Endowment_i$, from IR, 
    $O^f_i(\succ_i) = \{\Endowment_i\}$ holds, 
    since $\Endowment_i$ is the only house that satisfies IR. 
    Thus, $b^f_i(\succ_i;\succ_i)=w^f_i(\succ_i;\succ_i)=\Endowment_i$ holds. 
    Since $i$ always obtains the best house by truth-telling, there is no obvious 
    manipulation. 

    Assume $h^*\neq \Endowment_i$. Let $j\, (\neq i)$ be the initial owner
    of $h^*$, i.e., $\Endowment_j = h^*$ holds. 
    Let us construct $\succ_{-i}$ s.t., 
    $j$ prefers $\Endowment_i$ the most, and other houses (except $\Endowment_j$) are
    unacceptable. Also, for each agent $k \in N \setminus \set{i,j}$,
    she prefers her initial endowment the most, and other houses are unacceptable.
    Then, the only allocation that satisfies IR and PE is that $i$ and $j$ exchange their initial endowments, and other agents keep their initial endowments: the all-endowment allocation is IR but Pareto dominated by this swap (both $i$ and $j$ strictly gain), and any allocation moving a house of an agent $k \neq i,j$ violates IR since $e_k$ is $k$'s unique acceptable house. Thus, $h^* \in O^f_i(\succ_i)$ holds.
    Since $h^*$ is the best house for $i$, $b^f_i(\succ_i;\succ_i)=h^*$ holds. Then, no manipulation can 
    improve the best case. 
    
    Let us construct another $\succ_{-i}$ s.t., 
    for each agent $k (\neq i)$, 
    she prefers her initial endowment the most, and other houses are unacceptable.
    Under this $\succ_{-i}$, 
    the only allocation that satisfies IR is the one in which each agent gets her initial endowment regardless of the preference of agent $i$.
    Thus, $\Endowment_i \in O^f_i(\succ_i)$ and $\Endowment_i \in O^f_i(\succ'_i)$ hold, 
    where $\succ'_i$ is any possible manipulation. 
    Since $f$ is IR, each house in $O^f_i(\succ_i)$ must be acceptable for $i$.
    Thus, $w^f_i(\succ_i;\succ_i)=\Endowment_i$ holds. 
    Also, since $\Endowment_i \in O^f_i(\succ'_i)$, 
    $\Endowment_i = w^f_i(\succ_i;\succ_i) \succeq_i w^f_i(\succ_i;\succ'_i)$ holds.
\end{proof}

The above theorem implies that NOM is too weak as a reasonable requirement for
Shapley--Scarf markets; any mechanism that satisfies very basic requirements 
(i.e., IR and PE) automatically satisfies it. 
Thus, we introduce a stronger incentive property based on agent groups. 

\begin{definition}[Groups]
We assume agents are partitioned into disjoint groups
$G = \set{g_1, \ldots, g_{\ell}}$, 
where $g_k\cap g_{k'}=\emptyset$ for $k\neq k'$ and
$\bigcup_{g\in G}g=N$.
We refer to the elements of $G$ as base groups. 
\end{definition}
Fix $g \in G$ and $i \in g$. 
Let $\succ_{g-i}$ denote the preference profile of agents in $g$ except for $i$, and 
let $\succ_{-g}$ denote the preference profile of agents outside of group $g$.

For any $\succ_i$, $\succ_{g-i}$, where $\succ_g = (\succ_i, \succ_{g-i})$,
let $O^f_i(\succ_g)$ denote the set of houses that $i$ has a chance to obtain
under $\succ_g$ and mechanism $f$, i.e., 
$O^f_i(\succ_g) = \set{f_i(\succ_g, \succ_{-g}) \mid \ \succ_{-g} \in \Prefs^{N\setminus g}}$.
Furthermore, let $b^f_i(\succ'_i;\succ_g)$ denote the best house in 
$O^f_i(\succ_g)$ based on $\succ'_i$, i.e., $b^f_i(\succ'_i;\succ_g)$ $=h$ s.t.\  
$h\in O^f_i(\succ_g)$ and $h\succ'_i h'$ holds for any $h'\in O^f_i(\succ_g)\setminus 
\{h\}$. 
Similarly, let $w^f_i(\succ'_i;\succ_g)$ denote the worst house in 
$O^f_i(\succ_g)$ based on $\succ'_i$, i.e., $w^f_i(\succ'_i;\succ_g)=h$ s.t.\  
$h\in O^f_i(\succ_g)$ and $h'\succ'_i h$ holds for any $h'\in O^f_i(\succ_g)\setminus 
\{h\}$. These notations extend the corresponding notations for standard NOM. 

\begin{definition}[NOM with Groups]
A mechanism $f$ is not obviously manipulable with groups, or satisfies NOM with
groups, if $\forall g\in G$, $\forall i\in g$, 
$\forall \succ_i\in\Prefs$, $\forall \succ'_i\in\Prefs$, and $\forall \succ_{g-i}\in\Prefs^{g\setminus \set{i}}$,
\[
    b_i^f(\succ_i;(\succ_i,\succ_{g-i}))
    \succeq_i
    b_i^f(\succ_i; (\succ'_i,\succ_{g-i})),
\]
and
\[
    w_i^f(\succ_i;(\succ_i,\succ_{g-i}))
    \succeq_i
    w_i^f(\succ_i;(\succ'_i,\succ_{g-i}))
\] hold.
\end{definition}

In words, a mechanism satisfies NOM with groups if an agent can never strictly improve their best-case or worst-case outcomes by misreporting their preferences, assuming they know the exact reports of their own group members but remain completely uncertain about the reports of outsiders.

\begin{remark}
If $G=\{N\}$, then the reports outside the group are absent and the above
condition coincides with strategy-proofness. If every group is a singleton,
i.e., $\forall g \in G, |g|=1$ holds,  
the
condition coincides with the standard notion of non-obvious manipulability.
\end{remark}

\section{New Class of Mechanisms: TTC with super-groups}
\label{sec:supergroups}

\begin{definition}[Super-groups]
We say that a family of subsets of agents $\Supergroups \subseteq 2^{N}\setminus \{\emptyset\}$ is a family of super-groups if the following conditions hold:\\
    (i) for any $\ell \in \Supergroups$, there exists $G'\subseteq G$ s.t.\ 
    $\ell = \bigcup_{g \in G'} g$.\\
    (ii) for any $g \in G$, $g\in L$.\\
    (iii) $N\in \Supergroups$. 
\end{definition}

Super-groups allow the mechanism to merge base groups together in stages, expanding the pool of agents who can trade with one another until everyone is included.

\begin{definition}[TTC with super-groups]
    TTC with super-groups $\Supergroups$ is defined as follows:
    \begin{enumerate}
        \item If all elements in $\Supergroups$ are already selected, return the current allocation.
        \item Choose $\ell \in \Supergroups$ s.t.\ it is not selected yet, and there exists no $\ell' \in \Supergroups$ s.t.\ $\ell'$ is not selected yet and $\ell' \subsetneq \ell$. Apply TTC for agents in $\ell$, treating their currently held houses as their endowments. Go to 1. 
    \end{enumerate}
\end{definition}
\begin{remark}
If several super-groups can be chosen at Step~2, we can apply any tie-breaking rule, which is independent of agents'
declared preferences. Each choice of
$\Supergroups$ and a tie-breaking rule defines a deterministic TTC with super-groups mechanism. 
We use TTC with super-groups to refer to the class of all mechanisms obtained in this way. 
\end{remark}

Let us show an example. 
\begin{example}
\label{ex:ttc-with-super-groups}
Assume $N=\set{1,2,3}$, $G=\set{\{1,2\},\{3\}}$, $\Supergroups=\set{\{1,2\}, \{3\}, \{1,2,3\}}$, $e_i = h_i$ for each $i \in N$. 
Suppose $h_3 \succ_1 h_2 \succ_1 h_1$, $h_3 \succ_2 h_1 \succ_2 h_2$, 
and $h_2 \succ_3 h_1 \succ_3 h_3$. 

First, the mechanism applies TTC for $\{1,2\}$. Then, 
they swap their initial endowments. 
By applying TTC to $\{3\}$, 
agent~3 keeps her initial endowment. 
Then, the mechanism applies TTC globally. 
Agents 1 and 3 swap their current endowments. 
The final allocation is: $(h_3, h_1, h_2)$. 
\end{example}

\begin{remark}
TTC with super-groups is not strategy-proof. 
In Example~\ref{ex:ttc-with-super-groups}, 
assume agent~2  misreports her preference as:
$h_3 \succ'_2 h_2 \succ'_2 h_1$. Then, there is no swap between agents 1 and 2 
when applying TTC to $\{1, 2\}$.
The final allocation is: $(h_1, h_3, h_2)$. 
Thus, agent 2 obtains a better house by this misreport. 
\end{remark}

Let us introduce another participation condition, which we call 
group rationality. 
\begin{definition}[Group Rationality (GR)]
    For group $g$, let $c(\succ_g)$ denote the core allocation when only agents in $g$ are in the market. 
    We say that a mechanism $f$ satisfies GR if for every $\succ \in \Prefs^N$, every $g \in G$, and every $i \in g$, $f_i(\succ) \succeq_i c_i(\succ_g)$ holds.
\end{definition}

\begin{remark}
    In the Shapley--Scarf market, the core is unique and identical to the TTC allocation \citep{roth1977weak, shapley1974}.
\end{remark}
\begin{remark}
GR implies IR. 
\end{remark}

\begin{remark}
GR is a reasonable requirement when each group can decide whether to join the 
global integrated market, e.g., a hospital/region/country can decide 
whether to join the global organ exchange network \citep{ashlagiroth2014, afkp2015}.
\end{remark}

\begin{remark}
If $f$ satisfies NOM with groups, many variants of $f$ do as well. 
Fix $\succ_g$ and suppose that, as $\succ_{-g}$ varies, $f$ yields only two allocations $x$ and $y$ for the agents in $g$. Let us modify $f$ at a single profile so that it returns $y$ instead of $x$.
As long as every best and worst case is unchanged for every $i$ and $g$, 
the modified mechanism also satisfies NOM with groups.
\end{remark}

Since there exist many almost equivalent mechanisms, 
we classify mechanisms using the following equivalence relation. 

\begin{definition}[Envelope equivalence]
We say mechanisms $f$ and $f'$ are envelope equivalent 
if $\forall g$, $\forall i\in g$, $\forall \succ_{g} \in \Prefs^g$,
$b^f_i(\succ_i;\succ_{g}) = b^{f'}_i(\succ_i;\succ_{g})$ and 
$w^f_i(\succ_i;\succ_{g}) = w^{f'}_i(\succ_i;\succ_{g})$ 
hold.
\end{definition}

Simply put, two mechanisms are envelope equivalent if they offer every agent exactly the same best-case and worst-case scenarios for any given preference profile of their group, meaning the extreme boundaries of their opportunity sets remain identical even if the mechanisms allocate houses differently in intermediate cases.

\begin{remark} \label{rmk:envelope-not-nom}
Even if $f$ and $f'$ are envelope equivalent and $f$ satisfies NOM with groups, this is not sufficient to guarantee
that $f'$ also satisfies NOM with groups. For example, there can be a house $h$ s.t.\ $h \in O^{f'}_i((\succ_i,\succ_{g-i}))$ but 
$h \not\in O^f_i((\succ_i,\succ_{g-i}))$, where $h$ lies between the best and worst houses in $O^f_i((\succ_i,\succ_{g-i}))$.
Adding $h$ does not matter when agent $i$'s preference is $\succ_i$, but 
for another preference $\succ'_i$, $h$ can be better, and agent $i$ might have an incentive to declare $\succ_i$ when her 
true preference is $\succ'_i$ in mechanism $f'$, i.e., it is possible that 
$b^{f'}_i(\succ'_i; (\succ_i, \succ_{g-i})) \succ'_i b^{f'}_i(\succ'_i; (\succ'_i,\succ_{g-i}))$ holds.
\end{remark}

Let us define a set of houses called \emph{protected core houses}, which is useful 
to identify the best house obtained by each agent. 
\begin{definition}[Protected core houses]
\label{def:protected}
For $g \in G$ and $\succ_g$, let $P(\succ_g)$ denote the protected core houses, 
which are given by the fixed point of $P^r$ defined as follows. 
\begin{enumerate}
    \item $P^0 = \emptyset$.
    \item $P^{r+1}=P^r \cup \{c_i(\succ_g) \mid c_i(\succ_g) \text{ is } i\text{'s best house within } H\setminus P^r \}$.
\end{enumerate}
\end{definition}
Intuitively, these are the houses assigned to agents who receive their top available choice from their group's internal exchange, meaning they have absolutely no incentive to participate in further trades.
The following theorem shows envelope-equivalence of 
all TTC with super-groups mechanisms.
\begin{theorem}
\label{thm:envelope}
All TTC with super-groups mechanisms are envelope-equivalent. 
Furthermore, $w^f_i(\succ_i;\succ_g)=c_i(\succ_g)$ holds for any TTC with super-groups mechanism 
$f$, for any $g\in G$, $\succ_g \in \Prefs^g$, $i \in g$, $\succ_i \in \Prefs$
\end{theorem}
\begin{proof}
For fixed $g \in G$ and $\succ_g$, 
we first show that in any TTC with super-groups mechanism, 
for each house $h \in P(\succ_g)$, 
$h$ is obtained by agent $j$ s.t. $c_j(\succ_g) = h$ holds. 
When $g$ is selected, each agent obtains her core allocation. 
Also, in $P(\succ_g)$, there exists at least one agent who prefers her core allocation the most. 
Then, she never points to other houses; she keeps her core allocation. 
Then, there exists another agent who prefers her core allocation the most, except for the
above house. She also keeps her core allocation. By recursively applying the above argument, 
each house in $P(\succ_g)$ is never traded after the TTC within $g$. 

Let us examine the best case. 
For each $i\in g$, 
let us consider the following two cases:
(i) $c_i(\succ_g) \in P(\succ_g)$, and (ii)
$c_i(\succ_g) \not\in P(\succ_g)$.

For case (i), let $h^*$ denote $c_i(\succ_g)$. 
Then, in both mechanisms, she obtains $h^*$. If there exists $h \in H$ such that $h\succ_i h^*$ holds, 
then $h$ is also included in $P(\succ_g)$. 
Thus, she cannot obtain it regardless
of $\succ_{-g}$. Therefore, 
the best case coincides across mechanisms.

For case (ii), agent $i$ cannot obtain any house included in $P(\succ_g)$ in either mechanism, since each agent $j$, where $c_j(\succ_g) \in P(\succ_g)$, 
obtains $c_j(\succ_g)$ in the TTC within $g$ and does not trade afterward. Hence neither $O^f_i(\succ_g)$ nor $O^{f'}_i(\succ_g)$ contains a house in $P(\succ_g)$. Let $h^*$ denote the best house for $i$ under $\succ_i$ among those not in $P(\succ_g)$, and let $j$ be its initial owner. Since $h^*$ is by definition the best house outside $P(\succ_g)$, it follows that $h^* \succeq_i h$ for every $h \in O^f_i(\succ_g)$ and every $h \in O^{f'}_i(\succ_g)$. It therefore remains to show that $h^* \in O^f_i(\succ_g)$ and $h^* \in O^{f'}_i(\succ_g)$ both hold.

There are the following two cases: (a) $j$ is outside of $g$, or (b) 
$j$ is inside of $g$. 

For case (a), 
    let us assume, for each agent outside of $g$ (except for $j$), 
    her initial endowment is her only acceptable house.
    Also assume agent $j$ prefers 
    $c_i(\succ_g)$ over her initial endowment $e_j$ (and other houses are unacceptable).
    Then, 
    every outsider except for $j$ points only to her own house, so 
    trades occur only within $g$ before $i$ and $j$ are 
    selected together.
First, $i$ obtains $c_i(\succ_g)$ by TTC within $g$. 
Then, $i$ obtains $h^*$ when they are selected together, so $h^* \in O^f_i(\succ_g)$; the same construction applied to $f'$, gives $h^* \in O^{f'}_i(\succ_g)$.

For case (b), 
let us assume each agent outside of $g$ (except for $k$, which will be specified later) prefers her own initial endowment.
Trades only occur within $g$ before $i$ and $k$ are selected together.
Note that in any intermediate super-group stage that mixes 
$g$ with other base groups, 
every outsider present in that stage points to her own house, so no cross-group trades occur in those stages.
Then, consider the pointing path starting from $i$; it starts from $i$,
    and first points to the current owner of $h^*$ (it can be $j$ or another agent in $g$). Eventually, the path goes outside of $g$.
    If this does not happen, then there exists an exchange cycle within $g$ that strictly improves
    $c(\succ_g)$, which contradicts the fact that $c(\succ_g)$ is the core allocation within $g$.
    Let us choose $k$ as the first agent on the path who is outside of $g$. 
    Assume $k$ prefers $c_i(\succ_g)$ over her initial endowment
    (and other houses are unacceptable).
    Then, $i$ gets $h^*$, so $h^* \in O^f_i(\succ_g)$; applying the same construction to $f'$ gives $h^* \in O^{f'}_i(\succ_g)$.

In both cases, $h^*$ is exactly the best house $i$ can obtain, i.e., $h^* = b^f_i(\succ_i;\succ_g) = b^{f'}_i(\succ_i\nobreak;\succ_g\nobreak)$.

Next, let us examine the worst-case. 
The TTC within $g$ is first applied to agents in $g$,
and each agent $i$ obtains $c_i(\succ_g)$. Then, the assignment of agent $i$ weakly improves 
afterward. 
Thus, for each $h \in O^f_i(\succ_g)$, $h\succeq_i c_i(\succ_g)$ holds. 
Let us consider the case that for each agent $j \not\in g$, her initial endowment 
$e_j$ is the only acceptable house. 
Then, there will be no trade afterward. Thus, $c_i(\succ_g) \in O^f_i(\succ_g)$ holds. 
Thus, $w^f_i(\succ_i;\succ_g) = c_i(\succ_g)$ holds. 
\end{proof}

Next, we show the properties that any TTC with super-groups mechanism satisfies. 
\begin{theorem}
\label{thm:ttc-with-super-groups}
Any TTC with super-groups mechanism is GR, PE, and NOM with groups.    
\end{theorem}
\begin{proof}
Let $f$ be a TTC with super-groups mechanism. 
From Theorem~\ref{thm:envelope}, $c_i(\succ_g) =$ $w^f_i(\succ_i\nobreak;\succ_g\nobreak)$ holds. 
Thus, TTC with super-groups is GR. 
Also, since TTC with super-groups applies TTC globally in the last step, 
and TTC is PE, TTC with super-groups also satisfies PE.

We are going to show that TTC with super-groups satisfies NOM with groups. First, let us examine the worst case. As shown in the proof of Theorem~\ref{thm:envelope}, 
the TTC within $g$ is applied, and each agent $i$ first obtains $c_i(\succ_g)$. 
In the worst case, there will be no trade afterward.
Since TTC is strategy-proof, agent $i$ cannot improve her assignment in 
the TTC within $g$. Thus, no manipulation can improve this worst-case. 

Second, let us examine the best case. 
Let $h^*$ denote agent $i$'s best-case house under truthful reporting,
as characterized in Theorem~\ref{thm:envelope}.
Suppose that agent $i$ obtains $h$ such that $h\succ_i h^*$ by misreporting. 
By the definition of $h^*$, such a house must belong to
$P(\succ_g)$.
However, the construction of $P(\succ_g)$ implies that a protected
core house that is strictly preferred by $i$ to $h^*$ cannot become
available to $i$ through a change in $i$'s report (otherwise, there exists a profitable manipulation in TTC).
Hence no misreport can improve the best case.
\end{proof}

We show the envelope equivalence of all mechanisms 
satisfying GR, PE, and NOM with groups. 
\begin{theorem} \label{thm:characterization}
If a mechanism $f$ satisfies GR, PE, and NOM with groups, then $f$ 
is envelope-equivalent to every 
TTC with super-groups mechanism. 
\end{theorem}
\begin{proof}
It is sufficient to show that for every mechanism $f$ satisfying GR, PE, and 
NOM with groups, 
$w^f_i(\succ_i;\succ_g) = c_i(\succ_g)$ and $b^f_i(\succ_i;\succ_g)=h^*$,
where $h^*$ is $i$'s best house within $H\setminus P(\succ_g)$ when 
$c_i(\succ_g) \not\in P(\succ_g)$. Otherwise, $h^* = c_i(\succ_g)$.

For the worst case, assume each agent outside of $g$ considers only her own initial endowment acceptable. 
Then, allocating $c(\succ_g)$ for agents within group $g$ is the only GR outcome. 
Thus, $c_i(\succ_g) \in O^f_i(\succ_g)$ holds. 
Also, from GR, any outcome in $O^f_i(\succ_g)$ must be weakly better than $c_i(\succ_g)$. Thus, $w^f_i(\succ_i;\succ_g) = c_i(\succ_g)$ holds. 

For the best case, first note that GR fixes every protected core
house to its core recipient. Indeed, this follows by induction on
the construction of $P(\succ_g)$: each of the first protected agents ranks
her core house first among all houses, and after the previously
protected houses have been fixed, the same argument applies to the
next protected agents.
If $c_i(\succ_g)\in P(\succ_g)$, then
$h^*=c_i(\succ_g)$, and GR implies that agent $i$ always obtains
$h^*$. Hence the desired best-case equality holds.

Suppose, therefore, that
$c_i(\succ_g)\notin P(\succ_g)$.
Then, agent $i$ cannot obtain any house in $P(\succ_g)$ under GR.
Hence
\[
b_i^f(\succ_i;\succ_g)\preceq_i h^*,
\]
where $h^*$ is the most preferred house of $i$ in
$H\setminus P(\succ_g)$.
Moreover, by the definition of $P(\succ_g)$,
\[
h^*\succ_i c_i(\succ_g).
\]

It remains to show that $h^*\in O_i^f(\succ_g)$.
Suppose, for contradiction, that
$h^*\notin O_i^f(\succ_g)$.
Define $\succ'_i$ by ranking $h^*$ first and
$c_i(\succ_g)$ second, while preserving the relative order of all
remaining houses.
For $\succ'_g = (\succ'_i, \succ_{g-i})$, $c(\succ'_g) = c(\succ_g)$ holds.\footnote{Note that $h^* \succ_i c_i(\succ_g)$: the worst-case argument gives $c_i(\succ_g) \in O^f_i(\succ_g)$, we have shown $h^* \succeq_i h$ for every $h \in O^f_i(\succ_g)$, and $h^* \neq c_i(\succ_g)$. Hence $\set{h : h \succ'_i c_i(\succ_g)} = \set{h^*} \subseteq \set{h : h \succ_i c_i(\succ_g)}$, so any coalition $S \subseteq g$ blocking $c(\succ_g)$ under $\succ'_g$ would also block it under $\succ_g$; since the core is unique, $c(\succ'_g) = c(\succ_g)$ holds.}

Let $a_0=i$. For each agent $a_r\in g$ such that
$c_{a_r}(\succ_g)\notin P(\succ_g)$, let $q_{a_r}$ denote her
most preferred house in $H\setminus P(\succ_g)$.
By the definition of $P(\succ_g)$,
\[
q_{a_r}\succ_{a_r} c_{a_r}(\succ_g).
\]
In particular, $q_i=h^*$.

If $q_{a_r}$ is initially owned by an agent outside of $g$, stop
the path and let $t$ denote its initial owner.
Otherwise, since $c(\succ_g)$ is an allocation of the houses
initially owned by agents in $g$, there exists a unique
$a_{r+1}\in g$ such that
\[
q_{a_r}=c_{a_{r+1}}(\succ_g).
\]
Moreover, $c_{a_{r+1}}(\succ_g)\notin P(\succ_g)$.
Repeating this procedure must eventually lead to an agent
$t\notin g$. Otherwise, the path would contain a cycle within $g$,
and assigning each agent $a_r$ on this cycle the house $q_{a_r}$
would strictly improve every agent on the cycle relative to
$c(\succ_g)$, contradicting the Pareto efficiency of the core
allocation.

Choose $\succ_t$ such that $c_i(\succ_g)$ is her most preferred
house and her initial endowment $\Endowment_t$ is second.
For every other agent outside of $g$, let her initial endowment be
her only acceptable house.

Suppose that $f_i(\succ'_g,\succ_{-g})\neq h^*$.
Since $h^*$ and $c_i(\succ_g)$ are respectively the first and
second choices of $i$ under $\succ'_i$, GR implies
\[
f_i(\succ'_g,\succ_{-g})=c_i(\succ_g).
\]
Hence agent $t$ cannot obtain $c_i(\succ_g)$ and, by GR, obtains
$\Endowment_t$. Every other agent outside of $g$ also obtains her
initial endowment. Therefore all houses initially owned by agents
in $g$ remain within $g$. Since every agent in $g$ must weakly
prefer her assignment to $c(\succ_g)$ and $c(\succ_g)$ is Pareto
efficient within $g$, the allocation within $g$ must be
$c(\succ_g)$.

However, consider the reallocation along the path
\[
i=a_0,a_1,\ldots,a_m,t.
\]
Assign $q_{a_r}$ to each $a_r$ and assign $c_i(\succ_g)$ to $t$.
Every agent on the path strictly improves, while all other agents
keep the same house. This contradicts PE. Hence
\[
f_i(\succ'_g,\succ_{-g})=h^*,
\]
and therefore $h^*\in O_i^f(\succ'_g)$.

Here, $\succ'_g = (\succ'_i, \succ_{g-i})$ differs from $\succ_g$ only in agent $i$'s declared preference. 
Since we assume $h^* \notin O^f_i(\succ_g)$, 
$b^f_i(\succ_i;\succ_g) \prec_i h^*$ holds. 
Also, since $h^* \in O^f_i(\succ'_g)$, 
$b^f_i(\succ_i; \succ'_g) \succeq_i h^*$ holds. 
Hence $b^f_i(\succ_i;\succ'_g) \succ_i b^f_i(\succ_i;\succ_g)$: declaring $\succ'_i$ strictly improves agent $i$'s best case under her true preference $\succ_i$, contradicting NOM with groups.
\end{proof}

\begin{remark}
\autoref{thm:characterization} states a necessary condition: every mechanism satisfying GR, PE, and NOM with groups is envelope-equivalent to every 
TTC with super-groups mechanism. The converse does not hold. By Remark~\ref{rmk:envelope-not-nom}, a mechanism sharing this envelope need not itself satisfy NOM with groups.
\end{remark}

\section{Conclusions}
\label{sec:conclusions}
We studied how the Shapley--Scarf housing market changes when strategy-proofness is replaced by a weaker, best- and worst-case incentive requirement. Ordinary NOM turns out to be too weak in this domain: IR and PE already imply it. This motivates NOM with groups, in which agents evaluate misreports by varying only the preferences of agents outside their group. The condition ranges from ordinary NOM, when groups are singletons, to strategy-proofness, when the whole market is one group.

We also added GR, requiring every group to be protected by its own internal core allocation. 
We then introduced TTC with super-groups as a family of mechanisms that integrate groups in stages. These mechanisms satisfy GR, PE, and NOM with groups.
All TTC with super-groups mechanisms have the same envelope, and our main result shows that any mechanism satisfying GR, PE, and NOM with groups must induce exactly this envelope.

The characterization is deliberately at the level of envelopes. NOM with groups depends only on best and worst outcomes, so it cannot identify the full allocation rule on profiles that generate intermediate outcomes. This leaves several directions for future work. One is to add further axioms that distinguish among mechanisms sharing the same envelope. Another is to study randomized mechanisms, weak preferences, or richer endowment structures. Finally, the connection to kidney exchange suggests studying staged market integration under additional constraints such as cycle-length limits, compatibility restrictions, or dynamic participation.

\subsubsection*{Acknowledgments}
Initial ideas for this work were obtained during 
the group work at 
Fair and Explainable Collective Decision Workshop 2026.
The authors are grateful to the workshop organizers,
Institut Pascal at Université Paris-Saclay,
and 
the program ``Investissements d’avenir'' ANR-11-IDEX-0003-01.
Yokoo is partially supported by JST ERATO Grant Number JPMJER2301 and
JSPS KAKENHI Grant Number 25K03186. 
\newpage
\bibliographystyle{ecta}
\bibliography{references}
\end{document}